\documentclass[11pt,a4paper]{article}
\usepackage[margin=1in]{geometry}
\usepackage[T1]{fontenc}
\usepackage{lmodern,microtype}
\usepackage{amsmath,amssymb,amsthm,mathtools,booktabs,array}
\usepackage{placeins,needspace}
\usepackage{tikz}
\usepackage{algorithm}
\usepackage[noend]{algpseudocode}
\usetikzlibrary{arrows.meta,positioning,calc,decorations.pathreplacing,fit,backgrounds}
\usepackage{url}
\usepackage[colorlinks,linkcolor=blue!55!black,citecolor=blue!55!black,urlcolor=blue!55!black]{hyperref}
\newtheorem{theorem}{Theorem}[section]
\newtheorem{lemma}[theorem]{Lemma}
\newtheorem{proposition}[theorem]{Proposition}

\theoremstyle{definition}

\newcommand{\OPT}{\operatorname{OPT}}
\newcommand{\poly}{\operatorname{poly}}
\newcommand{\one}{\mathbf 1}
\newcommand{\ceil}[1]{\left\lceil #1\right\rceil}
\newcommand{\floor}[1]{\left\lfloor #1\right\rfloor}
\title{An EPTAS for Vector Scheduling with Time Intervals}
\author{Junho Hwang\\Independent Researcher}
\date{}
\hypersetup{pdftitle={An EPTAS for Vector Scheduling with Time Intervals},pdfauthor={Junho Hwang}}
\begin{document}
\maketitle
\begin{abstract}
We study vector scheduling in which each job is active during a fixed time
interval. A job uses several resources and stays on one machine for its
entire interval; its resource requirements may depend on the machine. The
objective is to minimize the largest resource load over all machines and
times. For $r$ machines and $d$ resources, we give a deterministic
$(1+\varepsilon)$-approximation in $f(r,d,1/\varepsilon)N^{O(1)}$ time,
where $N$ is the binary input length. This gives an efficient
polynomial-time approximation scheme for fixed $r$ and $d$, extending
approximation schemes for scalar temporary tasks assignment. The algorithm
merges jobs into blocks whose time intervals are fixed before any machine
is chosen, and assigns the blocks by dynamic programming over a balanced
recursive split of the time line. We also prove strong NP-hardness and an
exponential lower bound in $1/\varepsilon$ under the Exponential Time
Hypothesis, already for two identical machines and one resource.

\end{abstract}
\section{Introduction}\label{sec:introduction}
Vector scheduling models jobs that use several resources, such as CPU and
memory~\cite{ChekuriKhanna2004,BansalOosterwijkVredeveldZwaan2016}.
When each job reserves those resources for a specified time interval,
its contribution begins at arrival and ends at departure. We assign every
job to one machine for its entire interval and minimize the largest load
of any resource at any time. Resource requirements may depend on the
machine, and jobs may have specified sets of eligible machines.
Figure~\ref{fig:model} gives a two-resource example.

The single-resource case is \emph{offline temporary tasks
assignment}~\cite{AzarRegev1999,Azar2002,Armon2003}. Azar, Regev, Sgall,
and Woeginger (ARSW)~\cite[Theorem~4.4]{Azar2002} gave a polynomial-time
approximation scheme (PTAS) on $r$ identical machines, with running time
$n^{O(r^3\log r/\varepsilon^3)}$ for $n$ jobs. Armon et
al.~\cite[Section~4.1]{Armon2003} obtained a PTAS on a fixed number of
unrelated machines. We give an efficient PTAS (EPTAS) for the multi-resource
model: the exponent of the input length is independent of the machine
count, the number of resources, and the accuracy.

\begin{theorem}\label{thm:main}
For every $r,d\ge1$ and rational $0<\varepsilon\le1$, vector scheduling
with fixed time intervals on $r$ unrelated machines and $d$ resources
admits a deterministic $(1+\varepsilon)$-approximation in
$f(r,d,1/\varepsilon)N^{c_0}$ bit operations. Here $N$ is the binary input
length, $f$ is computable, and $c_0$ is an absolute constant. Arbitrary
nonnegative rational demands and nonempty eligibility sets are allowed.
\end{theorem}

Splitting the time line at all job endpoints turns an instance into
ordinary vector scheduling with one coordinate for every resource and
every piece of the time line. The number of pieces grows with the input,
and Theorem~\ref{thm:main} keeps it outside the parameter dependence. When
all jobs share the same interval, the model is ordinary vector scheduling
on unrelated machines.

The scalar identical-machine case has a sharper dependence on the
accuracy (Theorem~\ref{thm:scalar}). Even for two identical machines and
one resource, we prove strong NP-hardness and an exponential lower bound
in $1/\varepsilon$ under the Exponential Time Hypothesis
(Proposition~\ref{prop:lower-bounds}).

\paragraph*{Approach.}
All main ideas already appear for identical machines and one resource;
Sections~\ref{sec:blocks}--\ref{sec:analysis} treat this case and prove
Theorem~\ref{thm:scalar}. We split the time line recursively and collect
at each split the jobs that are active there. Those jobs are merged into
few blocks, each placed on a single machine. A block acts as one job whose
interval is the union of the intervals of its jobs. The blocks are chosen
so that this over-estimates the total load only slightly, and their
intervals are known before any machine is chosen. Each split cuts the
current part of the time line at its middle and at the median endpoint of
the earlier blocks inside it, so few block intervals end inside any part.
A dynamic program over the splits then assigns the blocks exactly: inside
a part, it remembers only the machines of the blocks whose intervals end
there and one constant load vector for the others.
Section~\ref{sec:general} extends the algorithm to several resources,
machine-dependent demands, and eligibility sets, and proves
Theorem~\ref{thm:main}. Each job gets a scalar weight, and its resource
requirements are rounded into a bounded number of types; within a type,
all requirements are fixed multiples of the scalar weight. Blocks are
formed within each type, and one dynamic program assigns the blocks of all
types together.

\paragraph*{Related work.}
The blocks follow the rounding phase of ARSW~\cite[Section~3]{Azar2002}.
Over-estimating loads by step functions with few steps, and the cut at the
median endpoint, follow Grandoni, M\"omke, and
Wiese~\cite[Sections~2.1--2.2]{GrandoniMomkeWiese2021}. Fixing the block
intervals before choosing machines permits the stronger running-time bound
here. Our rounding of resource requirements into types
(Section~\ref{sec:general}) extends the minimum-weight normalization and
grouping of Armon et al.~\cite[Section~4.1.2]{Armon2003} to several
resources.
For ordinary vector scheduling on identical machines, Chekuri and
Khanna~\cite{ChekuriKhanna2004} gave a PTAS in fixed dimension, and
Bansal et al.~\cite{BansalOosterwijkVredeveldZwaan2016} gave an EPTAS.
Their dimension parameter counts all coordinates; here the number of
time coordinates grows with the input. Temporal bin
packing~\cite{DellAmico2020} and Round-UFP~\cite{KarKhanWiese2022}
minimize the number of capacity-feasible bins or copies.
For scalar temporary tasks on identical machines, unit weights admit
an optimal polynomial-time assignment~\cite{Antoniadis2011}, while
variable weights rule out approximation ratios below $3/2$ in polynomial
time when $r$ is part of the input~\cite[Theorem~5.1]{Azar2002}.

\section{The model}\label{sec:preliminaries}
There are $n$ jobs, $r$ machines, and $d$ resources. Job $i$ has a fixed
half-open interval $I_i=[a_i,b_i)$ with $a_i<b_i$, a nonempty eligibility set
$E_i\subseteq[r]$, and a vector
$p_{ij}=(p_{ij1},\ldots,p_{ijd})\in\mathbb Q_{\ge0}^d$ for each
$j\in E_i$. An assignment $x$ chooses $x_i\in E_i$ once for each job.
Its load and peak are
\begin{equation}\label{eq:objective}
 L_{x,jk}(t)=\sum_{\substack{i:\,t\in I_i\\x_i=j}}p_{ijk},
 \qquad \max_{t,\,j\in[r],\,k\in[d]}L_{x,jk}(t).
\end{equation}
We write $\OPT$ for the minimum peak. A resource with a positive capacity
on machine $j$ is handled by dividing its demands on $j$ by that capacity.
The intervals remain fixed for every eligible machine.
All vector inequalities below are componentwise, and $\one$ is the all-ones
vector.

We assume that the $2n$ endpoints are distinct. Otherwise, sort the
endpoints by time, with departures before arrivals at equal times and
remaining ties broken by job index, and replace every endpoint by its rank.
Because departures come first, the jobs active between two consecutive
ranks are active together at some time of the original instance.
Conversely, the jobs active at an original time $t$ are active together in
the new instance right after the last endpoint at or before $t$. Hence
every assignment keeps its peak. The endpoints divide the span from the
first to the last endpoint into $2n-1$ half-open intervals, called
\emph{slabs} and numbered $1,\ldots,2n-1$ from left to right; loads are
constant on each slab. The \emph{arrival order}
lists jobs by increasing $a_i$, and the \emph{reverse-departure order} by
decreasing $b_i$.

All numerical input is rational and encoded in binary. The input length
$N$ includes the demands, endpoints, eligibility sets, $r,d$, and
$\varepsilon$. Running times count bit operations; every use of
$\poly(N)$ has an absolute degree. In Lemmas~\ref{lem:one-order},
\ref{lem:two-orders}, and~\ref{lem:merging}, polynomial time refers to the
length of the lemma's input, including its rational tolerances. Logarithms
are natural unless a base is specified.

\begin{figure}[tb]
\centering
\begin{tikzpicture}[x=.65cm,y=.48cm,font=\scriptsize,>=Stealth]
\node[anchor=west,font=\small\bfseries] at (-.45,6.55) {Jobs and resource demands};
\node at (6.8,4.8) {machine 1};
\node at (9.2,4.8) {machine 2};
\foreach \name/\a/\b/\yy/\vone/\vtwo in
 {A/0/4/3.6/{(2,1)}/{(1,3)},B/1/3/2.3/{(1,2)}/{(2,1)},C/3/5/1.0/{(2,1)}/{(1,2)}} {
 \node[anchor=east] at (-.3,\yy) {$\name$};
 \draw[line width=2.5pt,black!65] (\a,\yy)--(\b,\yy);
 \fill (\a,\yy) circle (1.3pt);
 \draw[fill=white] (\b,\yy) circle (1.3pt);
 \node at (6.8,\yy) {$\vone$};
 \node at (9.2,\yy) {$\vtwo$};
}
\draw[->] (-.15,0)--(5.6,0);
\foreach \xx in {0,1,2,3,4,5}{
 \draw (\xx,.12)--(\xx,-.12);
 \node[below] at (\xx,-.12) {$\xx$};
}
\node[anchor=west] at (5.7,0) {time};
\node at (7.9,-.8) {(CPU, memory)};
\begin{scope}[shift={(12,0)}]
\node[anchor=west,font=\small\bfseries] at (-.5,6.55) {Loads of one assignment};
\foreach \base/\mach in {3.1/1,0/2}{
 \draw[->] (0,\base)--(5.6,\base);
 \draw[->] (0,\base)--(0,{\base+2.5});
 \draw[black!25,dotted] (0,{\base+2})--(5.2,{\base+2});
 \node[anchor=east] at (-.12,{\base+2}) {$2$};
 \node[anchor=east] at (-.12,\base) {$0$};
 \node[anchor=west] at (3.1,{\base+2.45}) {machine \mach};
}
\draw[line width=.9pt] (0,5.1)--(4,5.1)--(4,3.1)--(5.2,3.1);
\draw[line width=.9pt,densely dashed] (0,4.1)--(4,4.1)--(4,3.1);
\draw[line width=.9pt] (0,0)--(1,0)--(1,2)--(3,2)--(3,1)--(5,1)--(5,0);
\draw[line width=.9pt,densely dashed] (0,0)--(1,0)--(1,1)--(3,1)--(3,2)--(5,2)--(5,0);
\foreach \xx in {0,1,2,3,4,5}{
 \draw (\xx,.10)--(\xx,-.10);
 \node[below] at (\xx,-.1) {$\xx$};
}
\draw[line width=.9pt] (.15,-1.2)--(.95,-1.2);
\node[anchor=west] at (1.1,-1.2) {CPU};
\draw[line width=.9pt,densely dashed] (2.55,-1.2)--(3.35,-1.2);
\node[anchor=west] at (3.5,-1.2) {memory};
\end{scope}
\end{tikzpicture}
\caption{Three interval jobs on two machines with two resources. Assigning
$A$ to machine~1 and $B,C$ to machine~2 gives peak load $2$.
The demand of each job depends on its assigned machine. The right panel
shows both resource loads over time; $B$ departs exactly when $C$ arrives.}
\label{fig:model}
\end{figure}
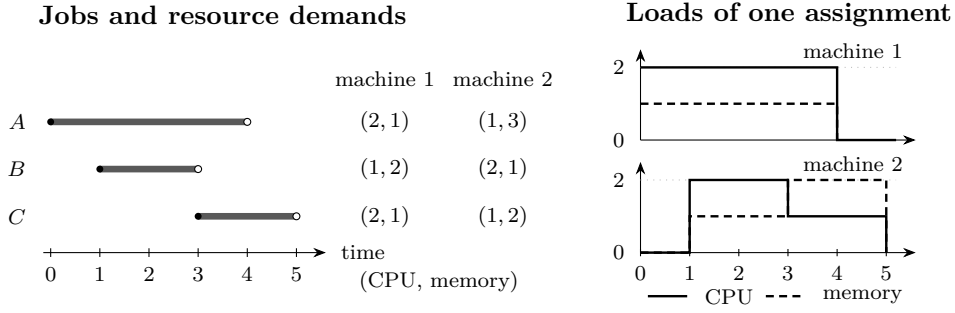

\paragraph*{Identical machines and one resource.}
Sections~\ref{sec:blocks}--\ref{sec:analysis} treat the case $d=1$ in which
every machine is eligible for every job and job $i$ has the same demand
$w_i\ge0$, its \emph{weight}, on every machine. This case contains all main ideas, and
Section~\ref{sec:general} extends the algorithm to the general model. A job
of weight zero can be placed on any machine, so we remove it and restore it
at the end; henceforth $w_i>0$. Empty instances and $r=1$ are solved
directly. The \emph{active weight} at time $t$ is
\[
 W(t)=\sum_{i:\,t\in I_i}w_i.
\]
Every assignment of peak at most $T$ has $W(t)\le rT$, since $W(t)$ is the
sum of the $r$ loads at $t$; placing all jobs on one machine gives peak
$\max_tW(t)$. Hence
\begin{equation}\label{eq:LB}
 \mathrm{LB}=\max\Bigl\{\max_iw_i,\ \frac1r\max_tW(t)\Bigr\}
\end{equation}
satisfies $\mathrm{LB}\le\OPT\le r\,\mathrm{LB}$.

\section{Blocks for two job orders}\label{sec:blocks}
A \emph{group} is a set of jobs that are all active on a specified common
slab $[t_-,t_+)$ and share a set $E$ of eligible machines; $E=[r]$ until
Section~\ref{sec:general}. Before $t_-$, its active jobs form a prefix of
its arrival order; from $t_+$ on, they form a prefix of its
reverse-departure order; on the common slab, all its jobs are active. We
partition a group into \emph{blocks} so that every assignment can be
imitated by one that keeps each block on a single machine, simultaneously
for all prefixes of both orders. The construction follows the rounding
phase of ARSW~\cite[Section~3]{Azar2002}.

Let $G$ be a set of jobs with weights $w_i>0$ and $E\subseteq[r]$ a
nonempty set of machines. For an assignment $x\in E^G$, a machine $j$,
and a set $F\subseteq G$, write
\[
 w^x_j(F)=\sum_{\substack{i\in F\\x_i=j}}w_i,\qquad
 w(F)=\sum_{i\in F}w_i.
\]

\begin{lemma}[One order]\label{lem:one-order}
For an order of $G$ and a rational \emph{tolerance} $\xi>0$, one can compute
in polynomial time a partition of $G$ into at most $1+8w(G)/\xi$ sets of
consecutive
jobs, each a single job or of weight less than $\xi/2$, with the following
property. For every $x\in E^G$ there is an assignment $y\in E^G$ that is
constant on every set and satisfies $w^y_j(F)<w^x_j(F)+\xi$ for every
prefix $F$ and every $j\in E$.
\end{lemma}
Jobs heavier than $\xi/4$ form sets of their own, and the other jobs are
cut greedily into sets of weight about $\xi/4$. Appendix~\ref{app:rounding}
gives the proof.

\begin{lemma}[Two orders]\label{lem:two-orders}
For two orders of a nonempty $G$ and a rational $\xi>0$, one can compute in
polynomial time a partition of $G$ into at most $(1+16w(G)/\xi)^2$ blocks
with the following two properties.
\begin{itemize}
\item[(a)] For every $x\in E^G$ there is an assignment $y\in E^G$ that is
constant on every block and satisfies $w^y_j(F)<w^x_j(F)+\xi$ for every
prefix $F$ of either order and every $j\in E$.
\item[(b)] For every prefix $F$ of either order, the blocks that meet $F$
without being contained in $F$ have total weight less than $\xi/2$.
\end{itemize}
\end{lemma}
\begin{proof}
Apply Lemma~\ref{lem:one-order} to the first order with tolerance $\xi$,
and let $G_1,\ldots,G_\mu$ be the resulting sets, where
$1\le\mu\le1+8w(G)/\xi$. Apply it again inside every $G_h$ to the
restriction of the second order, with tolerance $\xi/(2\mu)$; the resulting
sets are the blocks. The set $G_h$ yields at most $1+16\mu\,w(G_h)/\xi$
blocks, so there are at most $\mu(1+16w(G)/\xi)\le(1+16w(G)/\xi)^2$
blocks.

For~(a), fix $x$, and let $y$ agree on every $G_h$ with the assignment
that Lemma~\ref{lem:one-order} gives for the restriction of $x$ to $G_h$.
A prefix $F$ of the second order meets every $G_h$ in a prefix of its
restricted order, so $w^y_j(F)-w^x_j(F)<\mu\cdot\xi/(2\mu)=\xi/2$. A prefix
$F$ of the first order contains some sets $G_h$ completely, meets at most
one other set $G_{h'}$ partially, and misses the rest. Each complete set
$G_h$ is a prefix of its own restricted order, so the complete sets
contribute less than $\mu\cdot\xi/(2\mu)=\xi/2$ to $w^y_j(F)-w^x_j(F)$. The
set $G_{h'}$ is not a single job, so it contributes at most
$w(G_{h'})<\xi/2$.

For~(b), say that $F$ \emph{splits} a block if it meets the block without
containing it. A prefix of the first order splits only blocks inside
$G_{h'}$, of total weight less than $\xi/2$. A prefix of the second order
splits at most one block inside every $G_h$, since the blocks of $G_h$ are
consecutive in its restricted order. A split block is not a single job, so
it weighs less than $\xi/(4\mu)$, and all split blocks together weigh less
than $\xi/4$.
\end{proof}

\section{Merged jobs with fixed intervals}\label{sec:merged}
Let $G$ be a group with common slab $[t_-,t_+)$ and common eligible
machines $E$. For an assignment $x\in E^G$, a machine $j$, and a time $t$,
write
\[
 g_{x,j}(t)=\sum_{\substack{i\in G:\,t\in I_i\\x_i=j}}w_i,
 \qquad W_G(t)=\sum_{i\in G:\,t\in I_i}w_i.
\]
The intervals of the jobs of a block $S\subseteq G$ all contain the common
slab, so their union is the interval $[a_S,b_S)$ with
$a_S=\min_{i\in S}a_i$ and $b_S=\max_{i\in S}b_i$. The \emph{merged job}
of $S$ has this interval and a weight $\widehat w_S\ge w(S)$. An assignment
$y$ of the merged jobs chooses a machine $y_S\in E$ for every block; we
also write $y$ for the assignment that places every job of $S$ on $y_S$.
The \emph{merged load} of $y$ is
\[
 \widehat g_{y,j}(t)=\sum_{\substack{S:\,t\in[a_S,b_S)\\y_S=j}}\widehat w_S.
\]
The interval of a merged job is determined by its block alone, before any
machine is chosen. Figure~\ref{fig:merged} shows an example.

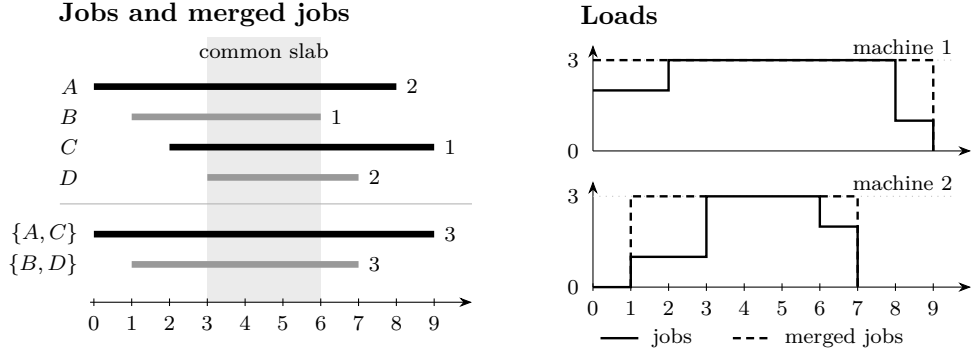
\begin{figure}[tb]
\centering
\begin{tikzpicture}[x=.5cm,y=.5cm,font=\scriptsize,>=Stealth,
 blockA/.style={black},blockB/.style={black!40}]
\node[anchor=west,font=\small\bfseries] at (-1.2,7.2) {Jobs and merged jobs};
\fill[black!8] (3,-.4) rectangle (6,6.6);
\node at (4.5,6.25) {common slab};
\def\job#1#2#3#4#5#6{\draw[#5,line width=2.5pt] (#2,#4)--(#3,#4);
 \node[anchor=east] at (-.15,#4) {$#1$};
 \node[anchor=west] at (#3,#4) {$#6$};}
\job{A}{0}{8}{5.3}{blockA}{2}
\job{B}{1}{6}{4.5}{blockB}{1}
\job{C}{2}{9}{3.7}{blockA}{1}
\job{D}{3}{7}{2.9}{blockB}{2}
\draw[black!30] (-.9,2.2)--(10,2.2);
\job{\{A,C\}}{0}{9}{1.4}{blockA}{3}
\job{\{B,D\}}{1}{7}{.6}{blockB}{3}
\draw[->] (-.1,-.4)--(10,-.4);
\foreach \x in {0,...,9}{\draw (\x,-.3)--(\x,-.5);\node[below] at (\x,-.5) {$\x$};}
\begin{scope}[shift={(13.2,0)}]
\node[anchor=west,font=\small\bfseries] at (-.6,7.2) {Loads};
\foreach \base/\mach in {3.6/1,0/2}{
 \draw[->] (0,\base)--(10,\base);
 \draw[->] (0,\base)--(0,{\base+2.8});
 \draw[black!25,dotted] (0,{\base+2.4})--(9.6,{\base+2.4});
 \node[anchor=east] at (-.1,{\base+2.4}) {$3$};
 \node[anchor=east] at (-.1,\base) {$0$};
 \node[anchor=west] at (6.6,{\base+2.75}) {machine \mach};
}
\draw[line width=.9pt] (0,5.2)--(2,5.2)--(2,6)--(8,6)--(8,4.4)--(9,4.4)--(9,3.6);
\draw[line width=.9pt,densely dashed] (0,6)--(9,6)--(9,3.6);
\draw[line width=.9pt] (0,0)--(1,0)--(1,.8)--(3,.8)--(3,2.4)--(6,2.4)--(6,1.6)--(7,1.6)--(7,0);
\draw[line width=.9pt,densely dashed] (1,0)--(1,2.4)--(3,2.4) (6,2.4)--(7,2.4)--(7,0);
\foreach \x in {0,...,9}{\draw (\x,.1)--(\x,-.1);\node[below] at (\x,-.1) {$\x$};}
\draw[line width=.9pt] (.2,-1.35)--(1.2,-1.35);
\node[anchor=west] at (1.3,-1.35) {jobs};
\draw[line width=.9pt,densely dashed] (3.7,-1.35)--(4.7,-1.35);
\node[anchor=west] at (4.8,-1.35) {merged jobs};
\end{scope}
\end{tikzpicture}
\caption{Merging blocks. Jobs $A,B,C,D$ have intervals
$[0,8),[1,6),[2,9),[3,7)$ and weights $2,1,1,2$; all are active on the
common slab $[3,6)$. With tolerance $\xi=24$ and the arrival order $A,B,C,D$ as
the first order, Lemma~\ref{lem:two-orders} keeps all four jobs in one set
and cuts the reverse-departure order $C,A,D,B$ at weight $3$. This gives
the blocks $\{A,C\}$ (black) and $\{B,D\}$ (gray), whose merged jobs have
intervals $[0,9)$ and $[1,7)$ and weight $3$. On the right, the blocks are
on machines 1 and 2. The merged load exceeds the load of the jobs only
while a block is partly active.}\label{fig:merged}
\end{figure}

\begin{lemma}[Merging a group]\label{lem:merging}
Let $G$ be a group, let $0<\eta\le1$ and $\tau>0$ be rational, and let $Q$
be an integer with
\begin{equation}\label{eq:block-count}
 Q\ge\ceil{(1+128/\eta)^2}\bigl(1+\floor{\log_2(1+\eta w(G)/\tau)}\bigr).
\end{equation}
One can compute in polynomial time a partition of $G$ into at most $Q$
blocks, and block weights $\widehat w_S\ge w(S)$ that are multiples of
$\tau/(4Q)$, with the following properties.
\begin{itemize}
\item[(a)] Every assignment $y$ of the merged jobs satisfies, for all
times $t$ and all $j\in E$,
\begin{equation}\label{eq:merged-upper}
 g_{y,j}(t)\le\widehat g_{y,j}(t).
\end{equation}
\item[(b)] For every $x\in E^G$, some assignment $y$ of the merged jobs
satisfies, for all times $t$ and all $j\in E$,
\begin{equation}\label{eq:merged-comparison}
 \widehat g_{y,j}(t)\le g_{x,j}(t)+2\bigl(\eta W_G(t)+\tau\bigr).
\end{equation}
\end{itemize}
\end{lemma}
The empty group has no blocks. The next three steps prove the lemma for a
nonempty group.

\paragraph*{Weight classes.}
The function $W_G$ is nondecreasing before the common slab and
nonincreasing afterwards. For a job $i\in G$, let $m_i$ be the smaller of
$W_G(a_i)$ and the value of $W_G$ on the slab that ends at $b_i$. These
are the weights of the prefix of the arrival order that ends with $i$ and
of the prefix of the reverse-departure order that ends with $i$. Put
\[
 \mathcal C_\ell=\{i\in G:2^\ell\tau\le\tau+\eta m_i<2^{\ell+1}\tau\},
 \qquad \sigma_\ell=2^{\ell-1}\tau.
\]
Only the indices $0\le\ell\le\log_2(1+\eta w(G)/\tau)$ occur. Call
$\mathcal C_\ell$ a \emph{weight class}; it is \emph{active} at $t$ if one
of its jobs is active at $t$. Then
\begin{equation}\label{eq:class-bounds}
 w(\mathcal C_\ell)<\frac{2^{\ell+2}\tau}{\eta},\qquad
 \sum_{\ell:\,\mathcal C_\ell\text{ active at }t}\sigma_\ell
 \le\eta W_G(t)+\tau.
\end{equation}
For the first bound, a job $i\in\mathcal C_\ell$ has
$m_i<2^{\ell+1}\tau/\eta$, so the arrival prefix or the reverse-departure
prefix that ends with $i$ weighs less than $2^{\ell+1}\tau/\eta$. Hence
$\mathcal C_\ell$ lies in the union of the longest such prefix of each
order. For the second bound, a job $i$ that is active at $t$ has
$m_i\le W_G(t)$ by the monotonicity of $W_G$. If $\ell^*$ is the largest
index of an active class, the geometric sum gives
$\sum_{\ell=0}^{\ell^*}\sigma_\ell<2^{\ell^*}\tau\le\tau+\eta W_G(t)$. With
no active class, the sum is zero.

\paragraph*{Blocks.}
For every nonempty class $\mathcal C_\ell$, apply
Lemma~\ref{lem:two-orders} with the arrival order of $\mathcal C_\ell$ as
the first order, its reverse-departure order as the second, and tolerance
$\xi=\sigma_\ell$. Since $w(\mathcal C_\ell)/\sigma_\ell<8/\eta$, the class
has at most
$(1+128/\eta)^2$ blocks, so by~\eqref{eq:block-count} the group has at most
$Q$ blocks. Round the weight of every block of $\mathcal C_\ell$ up to a
multiple of $\sigma_\ell/(2Q)$. This unit equals $2^\ell\tau/(4Q)$, a
multiple of $\tau/(4Q)$.

\paragraph*{Loads.}
Fix a class $\mathcal C_\ell$ and a time $t$, and let $F$ be the set of jobs
of $\mathcal C_\ell$ active at $t$. The merged job of a block $S$ of
$\mathcal C_\ell$ is active at $t$ exactly when $S$ meets $F$. Indeed, if
$t<t_-$, then $b_S\ge t_+>t$, and $a_S\le t$ says that some job of $S$ has
arrived. If $t\ge t_+$, then $a_S\le t_-\le t$, and $b_S>t$ says that some
job of $S$ has not departed. On the common slab, every block meets
$F=\mathcal C_\ell$.

Let $y$ be an assignment of the merged jobs. The active merged jobs of
$\mathcal C_\ell$ on machine $j$ have total weight at least $w^y_j(F)$;
summing over the classes proves~\eqref{eq:merged-upper}. They exceed
$w^y_j(F)$ by at most the weight of the blocks that meet $F$ without being
contained in $F$, plus the rounding of at most $Q$ blocks. The set $F$ is
empty, a prefix of the arrival order, all of $\mathcal C_\ell$, or a
prefix of the reverse-departure order. Hence
Lemma~\ref{lem:two-orders}(b) bounds the first excess by
$\sigma_\ell/2$, and the rounding adds less than
$Q\cdot\sigma_\ell/(2Q)=\sigma_\ell/2$.

Given $x\in E^G$, choose $y$ on every class by
Lemma~\ref{lem:two-orders}(a) for the restriction of $x$; this one choice
serves all times and machines. Then $w^y_j(F)<w^x_j(F)+\sigma_\ell$, so the
merged jobs of $\mathcal C_\ell$ contribute less than
$w^x_j(F)+2\sigma_\ell$ to $\widehat g_{y,j}(t)$ when $\mathcal C_\ell$ is
active at $t$, and nothing otherwise. Summing over the classes and
using~\eqref{eq:class-bounds} proves~\eqref{eq:merged-comparison}.

\section{A dynamic program on the time line}\label{sec:dp}
Fix a target $T>0$. Divide all weights by $T$, and reuse their notation for
the normalized weights. Reject the target if $\max_t W(t)>r$. Otherwise
every group has total weight at most $r$, because all its jobs are active
on its common slab. Set
\begin{equation}\label{eq:parameters}
\begin{gathered}
 D=1+\ceil{\log_2(2n)},\qquad
 \eta=\frac{\varepsilon}{4r},\qquad
 \tau=\frac{\varepsilon}{8D},\\
 Q=\ceil{(1+128/\eta)^2}\bigl(1+\floor{\log_2(1+2D)}\bigr),\qquad
 \lambda=\frac{\tau}{4Q}.
\end{gathered}
\end{equation}
Since $\eta w(G)/\tau\le\eta r/\tau=2D$ for every group $G$, the integer
$Q$ satisfies~\eqref{eq:block-count}. So Lemma~\ref{lem:merging} gives every
group at most $Q=O(\eta^{-2}\log(2+D))$ merged jobs, whose weights are
multiples of $\lambda$. The parameters split the error budget $\varepsilon$.
Merging costs at most $2(\eta W_G(t)+\tau)$ for each group with an active
job. These groups partition the active jobs, so the relative terms add up
to at most $2\eta r=\varepsilon/2$. By Lemma~\ref{lem:tree}, at most $2D$
groups have an active job at any time, so the additive terms add up to at
most $4D\tau=\varepsilon/2$. The algorithm builds a tree of groups with their
merged jobs and then decides by dynamic programming whether all merged
jobs can be assigned so that every load stays at most $1+\varepsilon$.

\subsection{The separator tree}
A \emph{region} $P=[l,u]$ is a consecutive range of slab indices; its
\emph{internal jobs} are the jobs whose entire interval lies in these
slabs. The \emph{separator tree} has a region at each node, and its root
contains all slabs and all jobs. Every node removes one or two of its
slabs, forms groups from the internal jobs that meet them, and passes the
other internal jobs to at most three children. The merged jobs of the
groups at the ancestors of $P$ are the \emph{ancestor merged jobs} of $P$.
Each of them contains the common slab of its group, which lies outside
$P$. So its interval meets $P$ in a proper nonempty prefix of $P$, a
proper nonempty suffix of $P$, all of $P$, or nothing. In the first two
cases it \emph{ends inside} $P$, at a boundary $\zeta$ between slabs
$\zeta-1$ and $\zeta$ of $P$, where $l<\zeta\le u$. In the third case it
\emph{covers} $P$, and in the fourth it \emph{misses} $P$.

The node removes the middle slab $\floor{(l+u)/2}$ and the \emph{median
slab} $\zeta$, where $\zeta$ is the lower median of the boundaries at which
ancestor merged jobs end inside $P$, counted with multiplicity. If none ends
inside $P$, the middle slab also serves as the median slab. The remaining
slabs form at most three children. One group consists of the internal jobs
that meet the median slab, and a second group of those that meet the
middle slab but not the median slab. Discard empty groups, and let
$\mathcal G(P)$ be the at most two remaining ones. Every other internal job
lies in one child, since an interval that meets two children crosses a
removed slab. Compute the merged jobs of the new groups by
Lemma~\ref{lem:merging}, and recurse. The whole tree, with all merged jobs,
is fixed before any machine is chosen. Figure~\ref{fig:separator-tree}
shows an example.

\begin{lemma}[Tree bounds]\label{lem:tree}
The tree has at most $2n-1$ nodes and at most $D$ levels, and its groups
partition the jobs. At any time at most $2D$ groups have an active job,
all on one root-to-leaf path. At most $4Q$ ancestor merged jobs end
inside any region.
\end{lemma}
\begin{proof}
Removing the middle slab leaves at most half of the slabs of a region,
rounded down, in each child. Every node removes at least one slab, and
every slab is removed at exactly one node. This gives the depth and node
bounds. Every internal job joins a group at its node or passes to exactly
one child, so the groups partition the jobs. The jobs of a group lie
inside its region, and the regions of one level are disjoint, which gives
the bound on active groups.

Suppose that $M$ ancestor merged jobs end inside $P$, at boundaries with
lower median $\zeta$. A merged job that ends at boundary $\zeta$ or
$\zeta+1$ has its endpoint at the left or right edge of the removed slab
$\zeta$, so it covers or misses every child. At most $\floor{M/2}$ of the
others end at boundaries below $\zeta$, and they can end inside only
children left of slab $\zeta$; at most $\floor{M/2}$ end at boundaries
above $\zeta+1$, and they can end inside only children right of slab
$\zeta$. The at most two new groups add at most $2Q$ merged jobs. Thus at
most $M/2+2Q$ ancestor merged jobs end inside any child, and induction
from zero at the root gives $4Q$.
\end{proof}

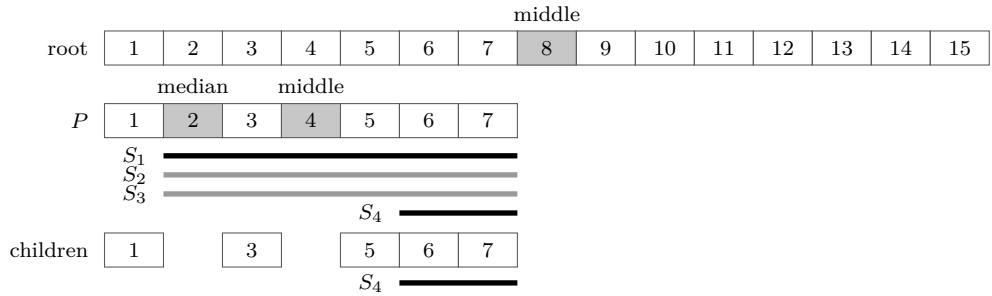
\begin{figure}[ht]
\centering
\begin{tikzpicture}[x=.78cm,y=.6cm,font=\scriptsize,
 mc1/.style={black},mc2/.style={black!40}]
\def\slab#1#2#3{\ifnum#3=1 \fill[black!22] ({#1-1},#2) rectangle (#1,{#2+.75});\fi
 \draw[black!65] ({#1-1},#2) rectangle (#1,{#2+.75});
 \node at ({#1-.5},{#2+.375}) {#1};}
\def\mjob#1#2#3#4#5{\draw[mc#5,line width=2pt] ({#2-1},#4)--(#3,#4);
 \node[anchor=east] at ({#2-1.1},#4) {$#1$};}
\node[anchor=east] at (-.1,.375) {root};
\foreach \i in {1,...,15}{\ifnum\i=8 \slab{\i}{0}{1}\else\slab{\i}{0}{0}\fi}
\node[anchor=south] at (7.5,.75) {middle};
\node[anchor=east] at (-.1,-1.225) {$P$};
\foreach \i in {1,3,5,6,7}{\slab{\i}{-1.6}{0}}
\foreach \i in {2,4}{\slab{\i}{-1.6}{1}}
\node[anchor=south] at (1.5,-.85) {median};
\node[anchor=south] at (3.5,-.85) {middle};
\mjob{S_1}{2}{7}{-2.0}{1}\mjob{S_2}{2}{7}{-2.42}{2}\mjob{S_3}{2}{7}{-2.84}{2}\mjob{S_4}{6}{7}{-3.26}{1}
\node[anchor=east] at (-.1,-4.075) {children};
\foreach \i in {1,3,5,6,7}{\slab{\i}{-4.45}{0}}
\mjob{S_4}{6}{7}{-4.8}{1}
\end{tikzpicture}
\caption{A separator tree on 15 slabs. The root has no ancestor merged
jobs and removes only its middle slab. The merged jobs $S_1$--$S_4$ of its
groups contain slab 8 and end inside its left child $P$; they are shown in
black on machine~1 and in gray on machine~2. Counted with multiplicity,
their boundaries $2,2,2,6$ have lower median 2, so $P$ removes slabs 2
and 4. Only $S_4$ ends inside a child of $P$.}
\label{fig:separator-tree}
\end{figure}

\subsection{States and reserved loads}\label{sec:states}
At a region $P$, all ancestor merged jobs are known, with their blocks,
weights, and intervals. A \emph{state} chooses a machine for every ancestor
merged job that ends inside $P$, and stores a constant load vector
$\kappa\in\lambda\mathbb Z_{\ge0}^{r}$. Its \emph{reserved load} on machine
$j$ is
\begin{equation}\label{eq:reserved}
 R_j(t)=\kappa_j
  +\sum_{\substack{S\text{ ends inside }P,\ t\in[a_S,b_S)\\
  S\text{ on machine }j}}\widehat w_S,
\end{equation}
where $S$ runs over the blocks of the ancestor groups. The vector $\kappa$
accounts for the ancestor merged jobs that cover $P$. A state with
$0\le\kappa\le(1+\varepsilon)\one$ and $R(t)\le(1+\varepsilon)\one$ on $P$ is
\emph{admissible}. Its table entry records whether the merged jobs of the
groups in the subtree of $P$ can be assigned so that $R$ plus their merged
loads stays at most $(1+\varepsilon)\one$ on $P$. At the root we query the
state $R=0$.

\subsection{Transitions by exact restriction}\label{sec:transitions}
For a state at $P$, try every assignment $y_G$ of the merged jobs of the
groups $G\in\mathcal G(P)$, and write $\widehat g_{y_G}(t)$ for the vector of
merged loads $(\widehat g_{y_G,j}(t))_j$. Test on all slabs of $P$ whether
\begin{equation}\label{eq:capacity}
 R(t)+\sum_{G\in\mathcal G(P)}\widehat g_{y_G}(t)\le(1+\varepsilon)\one.
\end{equation}
For each child $P'$, keep the machines of the merged jobs that end inside
$P'$, add the weight of every merged job that covers $P'$ to $\kappa$ on
its machine, and drop the merged jobs that miss $P'$. The ancestor merged
jobs of $P'$ are those of $P$ and those of the groups in $\mathcal G(P)$, so
the child state represents exactly
\begin{equation}\label{eq:exact-restriction}
 R'=\Bigl(R+\sum_{G\in\mathcal G(P)}\widehat g_{y_G}\Bigr)\Big|_{P'}.
\end{equation}
The new $\kappa$ lies in $\lambda\mathbb Z_{\ge0}^r$, because all merged
weights are multiples of $\lambda$. If~\eqref{eq:capacity} holds, then
$R'\le(1+\varepsilon)\one$ on $P'$ by~\eqref{eq:exact-restriction}, and the
new $\kappa$ is at most $R'(t)$ at every slab $t$ of $P'$, so the child
state is admissible. Accept the parent state if some tried assignment
passes~\eqref{eq:capacity} and all its child entries are accepted; at a
leaf, the capacity test decides. Store the choices and child pointers for
reconstruction.

In Figure~\ref{fig:separator-tree}, four ancestor merged jobs
$S_1,\ldots,S_4$ end inside $P=[1,7]$, at boundaries $2,2,2,6$, and all
four miss the child $[1,1]$. The state of $[3,3]$ adds the weights of
$S_1,S_2,S_3$ to $\kappa$. The state of $[5,7]$ adds them to $\kappa$ as
well and keeps the machine of $S_4$, which ends inside $[5,7]$.

\begin{lemma}[Table entries]\label{lem:table}
An admissible state at $P$ is accepted if and only if the merged jobs of
the groups in the subtree of $P$ have an assignment, with merged loads
$\widehat g_{y_G}$, such that
\begin{equation}\label{eq:table}
 R(t)+\sum_{G\text{ in the subtree of }P}\widehat g_{y_G}(t)
 \le(1+\varepsilon)\one\qquad(t\in P).
\end{equation}
\end{lemma}
\begin{proof}
Induct on the subtree. For an accepted state, the capacity test gives
\eqref{eq:table} on the removed slabs, where the merged jobs of
descendant groups are inactive. On a child, the induction hypothesis and
\eqref{eq:exact-restriction} give the same inequality. The removed slabs
and the children partition $P$.

Conversely, an assignment satisfying~\eqref{eq:table} passes the capacity
test because all loads are nonnegative. Its restricted child states are
admissible, and the remaining choices satisfy the child inequalities.
Induction accepts every child and then the parent. At a leaf, the
capacity test is the whole inequality.
\end{proof}

\section{Approximation and running time for one resource}\label{sec:analysis}
\subsection{The error bound}
\begin{lemma}[Target test]\label{lem:target}
Acceptance at target $T$ returns an assignment of peak at most
$(1+\varepsilon)T$. Every target $T\ge\OPT$ is accepted.
\end{lemma}
\begin{proof}
Use normalized units. For an accepted root, follow the stored choices to
an assignment of all merged jobs. By Lemma~\ref{lem:table}, their merged
loads are at most $(1+\varepsilon)\one$. Place every job on the machine of
its merged job. By~\eqref{eq:merged-upper}, its loads are at most the
merged loads. The groups partition the jobs, which proves the first
assertion.

For the second assertion, suppose $\OPT\le1$, and let $x$ be an assignment
of peak at most $1$. Then $W(t)\le r$, so the target is not rejected. In
every group choose the assignment $y_G$ of its merged jobs
that~\eqref{eq:merged-comparison} gives for the restriction of $x$. Because
the groups partition the jobs, $\sum_GW_G(t)=W(t)\le r$, and by
Lemma~\ref{lem:tree} at most $2D$ groups have an active job at $t$; the
merged jobs of the other groups are inactive at $t$. Writing $L_x(t)$ for
the vector of loads of $x$, we get
\begin{align}
 \sum_G \widehat g_{y_G}(t)
 &\le L_x(t)+2\sum_{G:\,W_G(t)>0}(\eta W_G(t)+\tau)\one\notag\\
 &\le\bigl(1+2\eta r+4D\tau\bigr)\one=(1+\varepsilon)\one.
 \label{eq:total-error}
\end{align}
By Lemma~\ref{lem:table}, the root is accepted.
\end{proof}

\subsection{Counting states and choices}\label{sec:running-time}
A node has at most two new groups with at most $Q$ merged jobs each, so it
tries at most $r^{2Q}$ assignments per state. By Lemma~\ref{lem:tree}, at
most $4Q$ ancestor merged jobs end inside a region, which allows at most
$r^{4Q}$ machine choices. Each of the $r$ coordinates of $\kappa$ has at
most $1+\floor{(1+\varepsilon)/\lambda}$ values. Thus the number of states
is at most
\begin{equation}\label{eq:state-count}
 r^{4Q}\left(1+\floor{\frac{1+\varepsilon}{\lambda}}\right)^r,
 \qquad \lambda=\frac{\varepsilon}{32DQ}.
\end{equation}
Let $A=\eta^{-2}\log(2r)+r$. Since $Q=O(\eta^{-2}\log(2+D))$, the machine
choices contribute at most $2^{O(A\log(2+D))}$, and the coordinates of
$\kappa$ contribute at most a factor depending only on $r$ and
$\varepsilon$ times $(2+D)^{O(r)}$. Each transition sweeps the slabs,
performs at most three child lookups, and updates machines and constant
loads by exact arithmetic. Even with linear table scans for lookups, only
a fixed power of these bounds is needed. Over at most $2n-1$ regions the
target test therefore takes
\begin{equation}\label{eq:time-before-absorption}
 (2+D)^{O(A)}\poly(N)
\end{equation}
bit operations, times a computable function of $r$ and $\varepsilon$.
Appendix~\ref{app:arithmetic} describes the exact arithmetic.

For every real $\omega\ge1$ and integer $n\ge1$,
\begin{equation}\label{eq:absorption}
 (1+\log n)^\omega\le\omega^\omega n.
\end{equation}
Indeed, maximizing $(1+z)^\omega e^{-z}$ over $z\ge0$ gives
$\omega^\omega e^{1-\omega}\le\omega^\omega$. Since $D=O(1+\log n)$,
\eqref{eq:absorption} absorbs the entire depth-dependent factor of
\eqref{eq:time-before-absorption} into a function of $r$ and
$\varepsilon$ times $n$, so the polynomial degree is absolute. Here
$\eta^{-1}=4r/\varepsilon$ and $A=O((r/\varepsilon)^2\log(2r))$, so this
function is
$2^{O(A\log(2+A))}=2^{O((r/\varepsilon)^2\log(2r)\log(2r/\varepsilon))}$.

\subsection{Choosing the target}
Algorithm~\ref{alg:scheme} summarizes the construction.
\begin{algorithm}[tb]
\caption{Approximation scheme for identical machines and one resource}
\label{alg:scheme}
\begin{algorithmic}[1]
\Require Jobs $(I_i,w_i)$, machine count $r$, accuracy $\varepsilon$
\State Make the endpoints distinct; remove zero-weight jobs; handle the empty instance and $r=1$.
\State Compute $\mathrm{LB}$.
\For{$\nu=0,\ldots,\ceil{(r-1)/\varepsilon}$}
  \State Set $T=\mathrm{LB}(1+\nu\varepsilon)$; use a fresh copy of the instance scaled by $1/T$.
  \If{$\max_t W(t)\le r$}
    \State Set the parameters~\eqref{eq:parameters}; build the separator tree and the merged jobs.
    \For{each region $P$ in bottom-up order}
      \State Enumerate admissible states and initialize their entries to false.
      \For{each state and assignment of the merged jobs of $\mathcal G(P)$ passing~\eqref{eq:capacity}}
        \State Compute each child state by~\eqref{eq:exact-restriction}.
        \If{all child entries are true}
          \State Mark the state true and store its choices and child pointers.
        \EndIf
      \EndFor
    \EndFor
    \If{the root state $R=0$ is true}
      \State Reconstruct its assignment, restore zero-weight jobs, and return.
    \EndIf
  \EndIf
\EndFor
\end{algorithmic}
\end{algorithm}

\begin{theorem}\label{thm:scalar}
Offline temporary tasks assignment on $r\ge2$ identical machines has
a deterministic $(1+\varepsilon)$-approximation in
\begin{equation}\label{eq:scalar-time}
 2^{O((r/\varepsilon)^2\log(2r)\log(2r/\varepsilon))}N^{O(1)}
\end{equation}
bit operations for rational $0<\varepsilon\le1$.
\end{theorem}
\begin{proof}
Test in increasing order
\begin{equation}\label{eq:target-grid}
 T_\nu=\mathrm{LB}(1+\nu\varepsilon),
 \qquad \nu=0,\ldots,\ceil{(r-1)/\varepsilon}.
\end{equation}
The grid reaches $r\,\mathrm{LB}\ge\OPT$. Its first target at least
$\OPT$ is at most $(1+\varepsilon)\OPT$, because the spacing is
$\varepsilon\,\mathrm{LB}\le\varepsilon\OPT$. That target is accepted by
Lemma~\ref{lem:target}. Stop at the first acceptance and reconstruct the
assignment. Its peak is at most
$(1+\varepsilon)^2\OPT\le(1+3\varepsilon)\OPT$, so running
Algorithm~\ref{alg:scheme} with $\varepsilon/3$ in place of $\varepsilon$
gives a $(1+\varepsilon)$-approximation. There are $O(r/\varepsilon)$
targets, each with polynomial encoding length, and
Section~\ref{sec:running-time} bounds each target test
by~\eqref{eq:scalar-time}.
\end{proof}

\section{Several resources and unrelated machines}\label{sec:general}
We now treat the general model and prove Theorem~\ref{thm:main}. Every job
gets a scalar weight, and its machine-dependent resource vectors are
rounded into a bounded number of types. Within a type, all resource
requirements are fixed multiples of the scalar weight. Sections~\ref{sec:blocks}
and~\ref{sec:merged} therefore apply within each type, and the dynamic
program of Sections~\ref{sec:dp} and~\ref{sec:analysis} assigns the merged
jobs of all types together, with the changes described below.

\subsection{A scalar weight for each job}
In the general model, the weight of job $i$ and the active weight are
\begin{equation}\label{eq:baseline}
 q_{ij}=\sum_{k=1}^d p_{ijk},\qquad
 w_i=\min_{j\in E_i}q_{ij},\qquad
 W(t)=\sum_{i:\,t\in I_i}w_i.
\end{equation}
For $d=1$ and identical machines, they are the $w_i$ and $W$ of
Section~\ref{sec:preliminaries}. Fix a minimizing machine $j_i^*$ for every
job. If $w_i=0$, assigning $i$ to $j_i^*$ contributes zero to every load,
so remove it and restore that choice at the end. We henceforth have
$w_i>0$. Empty instances and $r=1$ are solved directly. For any assignment
$x$ of peak at most $T$,
\begin{equation}\label{eq:baseline-bound}
 W(t)\le \sum_{i:\,t\in I_i}q_{i,x_i}
       =\sum_{j=1}^r\sum_{k=1}^d L_{x,jk}(t)\le rdT.
\end{equation}
Thus $W$ depends only on the input, and $\max_tW(t)\le rd\,\OPT$. In place
of~\eqref{eq:LB}, we use the computable lower bound
\begin{equation}\label{eq:LB-general}
 \mathrm{LB}=\max\Bigl\{\max_i\min_{j\in E_i}\|p_{ij}\|_\infty,\
 \frac1{rd}\max_t W(t)\Bigr\},
\end{equation}
which satisfies $\mathrm{LB}\le\OPT\le rd\,\mathrm{LB}$: the lower bound
follows from one job and~\eqref{eq:baseline-bound}, and assigning every job
to $j_i^*$ gives each resource load at most $W(t)$.

\subsection{Rounding the coefficients}
Fix a rational $0<\varepsilon\le1$ and put
\begin{equation}\label{eq:s}
 s=\ceil{rd/\varepsilon},\qquad
 E_i'=\{j\in E_i:q_{ij}\le s w_i\}.
\end{equation}
Every $E_i'$ contains $j_i^*$. For $j\in E_i'$ define
\begin{equation}\label{eq:round-coefficients}
 \alpha_{ijk}=\frac1s\ceil{\frac{s p_{ijk}}{w_i}},
 \qquad \widehat p_{ijk}=w_i\alpha_{ijk}.
\end{equation}
Each $\alpha_{ijk}$ belongs to $\{0,1/s,\ldots,s\}$, since
$p_{ijk}\le q_{ij}\le s w_i$. For bookkeeping, put $\alpha_{ijk}=0$ when
$j\notin E_i'$. The \emph{rounded instance} has the eligibility sets
$E_i'$ and the demands $\widehat p_{ijk}$; its loads are the \emph{rounded
loads}.

\begin{lemma}[Coefficient reduction]\label{lem:types}
If the original instance admits an assignment of peak at most $T$, the
rounded instance admits an assignment of peak at most
$(1+2\varepsilon)T$. Every assignment of the rounded instance has original
loads at most its rounded loads. The jobs have at most
\begin{equation}\label{eq:type-count}
 2^r(s^2+1)^{rd}
\end{equation}
different pairs $(E_i',(\alpha_{ijk})_{j,k})$.
\end{lemma}
\begin{proof}
Fix an original assignment $x$ of peak at most $T$. Move every job with
$x_i\notin E_i'$ to $j_i^*$, changing its machine for its whole
interval. At time $t$, the sum of $w_i$ over the moved active jobs is
at most
\[
 \frac1s\sum_{\substack{i:\,t\in I_i\\x_i\notin E_i'}}q_{i,x_i}
 \le\frac{rdT}{s}\le\varepsilon T.
\]
For each moved job, every demand on its new machine is at most $w_i$.
Hence this change increases any resource load by at most $\varepsilon T$.
The resulting assignment uses $E_i'$.

Rounding increases each demand by less than $w_i/s$. By
\eqref{eq:baseline-bound}, it therefore adds at most
$W(t)/s\le\varepsilon T$ to any resource load of this assignment. This
proves the first assertion. The second follows from
$\widehat p_{ijk}\ge p_{ijk}$ on retained choices. There are at most
$2^r$ eligibility sets and $s^2+1$ possibilities for each of the $rd$
coefficients, giving~\eqref{eq:type-count}.
\end{proof}

Call each distinct pair in Lemma~\ref{lem:types} a \emph{type}, and let $K$
be the number present. For a type $c$, write $E_c$ for its eligible
machines and $\alpha_{cjk}$ for its coefficients. A job of type $c$ has
rounded demand $w_i\alpha_{cjk}$, so an assignment's scalar load within one
type determines all its resource loads by multiplication. The types are
computed once and are unchanged by scaling all demands by a common target
value.

\subsection{Merged jobs of one type}\label{sec:merged-resources}
A group now consists of jobs of one type $c$, with eligible machines
$E=E_c$, and Lemma~\ref{lem:merging} applies to their scalar weights. A
merged job of weight $\widehat w_S$ on machine $j\in E_c$ has the resource
vector $\widehat w_S(\alpha_{cj1},\ldots,\alpha_{cjd})$ on that machine
during its interval. For an assignment $y$ of the merged jobs, let
$H_y(t)\in\mathbb Q_{\ge0}^{rd}$ be their resource loads, and let
$\widehat L_{x,G}(t)$ be the rounded resource loads of an assignment
$x\in E_c^G$. Since $0\le\alpha_{cjk}\le s$, Lemma~\ref{lem:merging} gives
\begin{align}
 \widehat L_{y,G}(t)&\le H_y(t),\label{eq:vector-domination}\\
 \forall x\in E_c^G\ \exists y:\quad
 H_y(t)&\le\widehat L_{x,G}(t)+2s\bigl(\eta W_G(t)+\tau\bigr)\one
 \label{eq:vector-comparison}
\end{align}
at all times, with the same $y$ for all $t$.

\subsection{The dynamic program and its analysis}
Fix a target $T>0$, divide $w_i$, $p_{ijk}$, and $\widehat p_{ijk}$ by
$T$, and reject the target if $\max_tW(t)>rd$; otherwise every group has
total weight at most $rd$. Keep $D$ from~\eqref{eq:parameters} and replace
its other parameters by
\begin{equation}\label{eq:parameters-general}
 \eta=\frac{\varepsilon}{4s\,rd},\quad
 \tau=\frac{\varepsilon}{8sKD},\quad
 Q=\ceil{(1+128/\eta)^2}\bigl(1+\floor{\log_2(1+2KD)}\bigr),\quad
 \lambda=\frac{\tau}{4Qs}.
\end{equation}
Then $\eta w(G)/\tau\le2KD$ for every group, so $Q$
satisfies~\eqref{eq:block-count}.

The separator tree is built as in Section~\ref{sec:dp}, except that every
node forms its two groups separately for each type, so a node has at most
$2K$ groups. The proof of Lemma~\ref{lem:tree} then shows that at most
$2KD$ groups have an active job at any time and that at most $4KQ$
ancestor merged jobs end inside any region. A state chooses a machine in
$E_c$ for every ancestor merged job of type $c$ that ends inside $P$, and
stores a constant load vector $\kappa\in\lambda\mathbb Z_{\ge0}^{rd}$. Its
reserved load is
\begin{equation}\label{eq:reserved-general}
 R_{jk}(t)=\kappa_{jk}
  +\sum_{\substack{S\text{ ends inside }P,\ t\in[a_S,b_S)\\
  S\text{ on machine }j}}\widehat w_S\,\alpha_{c(S)jk},
\end{equation}
where $c(S)$ is the type of the block $S$. Every resource demand of a
merged job is a multiple of $\lambda$, because $\widehat w_S$ is a multiple
of $\tau/(4Q)$ and every coefficient is a multiple of $1/s$.
Sections~\ref{sec:states} and~\ref{sec:transitions} and
Lemma~\ref{lem:table} carry over with three substitutions: the resource
loads $H_{y_G}$ replace the merged loads, a covering merged job $S$ on
machine $j$ adds $\widehat w_S\alpha_{c(S)jk}$ to $\kappa_{jk}$ for every
$k$ in place of its weight, and $1+3\varepsilon$ replaces
$1+\varepsilon$, also in the admissibility condition
$\kappa\le(1+3\varepsilon)\one$.

The error budget is now $3\varepsilon$. Rounding the coefficients costs
$2\varepsilon$ by Lemma~\ref{lem:types}. Merging costs at most
$2s(\eta W_G(t)+\tau)$ for each group with an active job; as in
Section~\ref{sec:dp}, these costs add up to at most
$2s\eta\,rd+4sKD\tau=\varepsilon$.

\begin{lemma}[Target test for several resources]\label{lem:target-general}
Acceptance at target $T$ returns an original assignment of peak at most
$(1+3\varepsilon)T$. Every target $T\ge\OPT$ is accepted.
\end{lemma}
\begin{proof}
Use normalized units, and argue as for Lemma~\ref{lem:target}: the groups
partition the jobs, the merged jobs of groups without an active job are
inactive, at most $2KD$ groups have one, and every chosen machine lies in
$E_c=E_i'\subseteq E_i$. For an accepted root, follow the stored choices
and place every job on the machine of its merged job. By
Lemma~\ref{lem:table} and~\eqref{eq:vector-domination}, its rounded loads
are at most $(1+3\varepsilon)\one$, and by Lemma~\ref{lem:types} its
original loads are at most its rounded loads. For the second assertion,
suppose $\OPT\le1$. Then $W(t)\le rd$ by~\eqref{eq:baseline-bound}, so the
target is not rejected, and Lemma~\ref{lem:types} gives a rounded
assignment $x$ of peak at most $1+2\varepsilon$. Choosing $y_G$
by~\eqref{eq:vector-comparison} in every group gives
\begin{align}
 \sum_GH_{y_G}(t)
 &\le\sum_G\widehat L_{x,G}(t)
  +2s\sum_{G:\,W_G(t)>0}(\eta W_G(t)+\tau)\one\notag\\
 &\le\bigl(1+2\varepsilon+2s\eta\,rd+4sKD\tau\bigr)\one
 =(1+3\varepsilon)\one,\label{eq:total-error-general}
\end{align}
so Lemma~\ref{lem:table} accepts the root.
\end{proof}

A node now tries at most $r^{2KQ}$ assignments per state, and the number of
states is at most
\[
 r^{4KQ}\left(1+\floor{\frac{1+3\varepsilon}{\lambda}}\right)^{rd},
 \qquad\lambda=\frac{\varepsilon}{32Qs^2KD}.
\]
Let $A$ now be $K\eta^{-2}\log(2r)+rd$. Since $Q=O(\eta^{-2}\log(2+KD))$
and $\log(2+KD)\le\log(2+K)+\log(2+D)$, the machine choices contribute at
most $2^{O(A\log(2+K))}(2+D)^{O(A)}$. The rest of the argument of
Section~\ref{sec:running-time} bounds a target test by
$(2+D)^{O(A)}\poly(N)$ bit operations times a computable function of $r$,
$d$, and $\varepsilon$, and~\eqref{eq:absorption} absorbs the
depth-dependent factor. For an explicit bound, $s=O(rd/\varepsilon)$,
$K\le2^r(s^2+1)^{rd}=(rd/\varepsilon)^{O(rd)}$, and
$\eta^{-1}=O((rd/\varepsilon)^2)$. Hence $A=(rd/\varepsilon)^{O(rd)}$, and
one target test takes
\begin{equation}\label{eq:general-time}
 2^{(rd/\varepsilon)^{O(rd)}}N^{O(1)}
\end{equation}
bit operations.

\begin{proof}[Proof of Theorem~\ref{thm:main}]
Compute the types once. Then proceed as in Algorithm~\ref{alg:scheme}, but
with $\mathrm{LB}$ from~\eqref{eq:LB-general}, the targets
$T_\nu=\mathrm{LB}(1+\nu\varepsilon)$ for
$\nu=0,\ldots,\ceil{(rd-1)/\varepsilon}$, rejection when
$\max_tW(t)>rd$, and the parameters~\eqref{eq:parameters-general}. The grid
reaches $rd\,\mathrm{LB}\ge\OPT$. So by Lemma~\ref{lem:target-general}, as
in the proof of Theorem~\ref{thm:scalar}, the first accepted target gives
peak at most $(1+3\varepsilon)(1+\varepsilon)\OPT\le(1+7\varepsilon)\OPT$.
Running the algorithm with $\varepsilon/7$ in place of $\varepsilon$ gives
a $(1+\varepsilon)$-approximation. There are $O(rd/\varepsilon)$ targets,
and the tree and the merged jobs are rebuilt for each of them.
So~\eqref{eq:general-time} gives the running time of
Theorem~\ref{thm:main} with $f(r,d,1/\varepsilon)=2^{(rd/\varepsilon)^{O(rd)}}$.
\end{proof}

\paragraph*{Further directions.}
Theorem~\ref{thm:scalar} and Proposition~\ref{prop:lower-bounds} leave the
dependence on the accuracy open. For several resources, reducing the number
of coefficient types would improve the function $f$ of
Theorem~\ref{thm:main}. Allowing a job to have several separate active
intervals introduces a further coupling: one job may meet several child
regions while avoiding their separating slabs.

\FloatBarrier
\section*{AI Usage Disclosure}
AI agents were used for exploratory discussions during the development of the
proofs and to assist with drafting and revising the manuscript. The author
takes full responsibility for the paper's contents and correctness.

\bibliographystyle{plainurl}
\bibliography{references}
\clearpage
\appendix
\section{The partition for one order}\label{app:rounding}
\begin{proof}[Proof of Lemma~\ref{lem:one-order}]
Make every job of weight greater than $\xi/4$ a set of its own. In each
maximal run of the other jobs, scan in order and close a set as soon as
its weight reaches $\xi/4$; keep a nonempty unfinished last set of the
run. A closed set weighs less than $\xi/2$, since it weighed less than
$\xi/4$ before its last job, which weighs at most $\xi/4$; an unfinished
set weighs less than $\xi/4$. The heavy jobs and the closed sets are disjoint and weigh at
least $\xi/4$ each, so there are at most $4w(G)/\xi$ of them. There is at
most one unfinished set per run and at most one more run than heavy jobs,
so there are at most $1+8w(G)/\xi$ sets.

Fix $x\in E^G$. We choose $y$ set by set, in the order. At the end of
each set, let $X_j$ and $Y_j$ be the weights that $x$ and $y$ have
assigned to $j$ so far; we keep
\[
 Y_j-X_j<\xi/2\qquad(j\in E).
\]
This holds initially. A single job keeps its machine from $x$. For a set
$B$ with several jobs, of weight $w(B)<\xi/2$, let $X_j^+$ include the
assignments of $x$ through the end of $B$, while $Y_j$ still stops before
$B$. Both assignments have placed the same weight before $B$, so
\[
 \sum_{j\in E}(X_j^+-Y_j)=w(B)>0.
\]
Choose a machine $j\in E$ with $Y_j<X_j^+$ and put all of $B$ there. Its
new difference is $Y_j+w(B)-X_j^+<w(B)<\xi/2$. On every other machine, $Y_j$ is
unchanged and $X_j$ does not decrease, so the bound holds at the end of
$B$.

A prefix that ends inside a set $B$ with several jobs adds less than
$\xi/2$ to the $y$-load of any machine after the end of the previous set,
and a nonnegative amount to its $x$-load. So $w^y_j(F)-w^x_j(F)<\xi$ for
every prefix $F$. The partition comes from one scan of exact prefix sums.
\end{proof}

\section{Exact arithmetic}\label{app:arithmetic}
All input rationals are written in binary, and the product of their
denominators has $O(N)$ bits. Every sum of input demands or scalar weights
is therefore a polynomial-bit numerator over a shared denominator, and
computing $q_{ij}$, $w_i$, $W$, and $\mathrm{LB}$ takes polynomially many
exact comparisons and additions. The integer $s=\ceil{rd/\varepsilon}$ has
$O(N)$ bits. A retained coefficient is stored as the integer
$\ceil{sp_{ijk}/w_i}\in\{0,\ldots,s^2\}$; comparing these integers and
the sets $E_i'$ identifies the types in polynomial time, and there are at
most $n$ of them. Every target $T_\nu$, in~\eqref{eq:target-grid} and in
the proof of Theorem~\ref{thm:main}, has polynomial bit length, and the
definitions of $\eta$, $\tau$, $Q$, $\lambda$, and the tolerances of
Lemma~\ref{lem:two-orders} add products and quotients of polynomial-bit
numbers.

For a group, the values $m_i$ are slab values of $W_G$, a class index
follows from comparisons with powers of two, and the partitions of
Lemmas~\ref{lem:one-order} and~\ref{lem:two-orders} come from scans of
exact prefix sums. A rounded block weight is an integer multiple of
$\tau/(4Q)$, and every demand of a merged job is an integer multiple of
$\lambda$. The dynamic program thus stores every coordinate of $\kappa$ as
an integer of at most $(1+\varepsilon)/\lambda$, or $(1+3\varepsilon)/\lambda$
for several resources, and restriction uses only integer addition.
Capacity tests, restrictions, and comparisons of states take polynomially
many operations on polynomial-bit numbers per slab, merged job, and
coordinate. Stored choices and child pointers reconstruct
one eligible machine for every job.

\section{Lower bounds for fixed machines}\label{app:hardness}
\newenvironment{labelledlist}{\begin{list}{}{\setlength{\leftmargin}{2.8em}%
  \setlength{\labelwidth}{2.4em}\setlength{\labelsep}{0.4em}%
  \setlength{\itemsep}{1pt}\setlength{\topsep}{3pt}}}{\end{list}}
\begin{proposition}\label{prop:lower-bounds}
For every fixed $r\ge2$, deciding whether the peak can be at most a given
threshold is strongly NP-complete, even with integer endpoints and weights
in $\{1,2\}$. Hence there is no FPTAS unless P$=$NP. Under the Exponential
Time Hypothesis (ETH), there is no deterministic
$(1+\varepsilon)$-approximation in $2^{o(1/\varepsilon)}\poly(N)$ time.
The analogous statement holds for randomized algorithms with success
probability at least $2/3$ under randomized ETH.
\end{proposition}
The construction uses scalar weights on identical machines. Throughout
this appendix $\theta$ denotes the threshold of the constructed instance.

\subsection{Balanced networks become interval instances}
A \emph{balanced network} is a directed acyclic multigraph with vertices
$v_0,\ldots,v_L$ in topological order, $L\ge1$, and positive integer arc
weights, such that every internal vertex $v_1,\ldots,v_{L-1}$ has equal
incoming and outgoing weight and the arcs leaving the source $v_0$ weigh
$2\theta$ in total. For $r\ge2$, its \emph{$r$-machine instance} has, for
each arc $(v_a,v_b)$ of weight $w$, one job of weight $w$ with interval
$[a,b)$, and in addition $(r-2)\theta$ jobs of weight $1$ with interval
$[0,L)$.

\begin{lemma}[Balanced networks]\label{lem:network}
In the $r$-machine instance of a balanced network, every slab $[k,k+1)$
with $0\le k<L$ has active weight $r\theta$, so $\OPT\ge\theta$. An
assignment of peak $\theta$ exists if and only if the arcs have a
coloring with $r$ colors in which every internal vertex has, for every
color, equal incoming and outgoing weight, and the arcs of each color
leaving $v_0$ weigh at most $\theta$. If no such coloring exists, then
$\OPT\ge\theta+1$.
\end{lemma}

\begin{proof}
The arcs active on $[k,k+1)$ are the arcs $(v_a,v_b)$ with $a\le k<b$, that
is, the arcs crossing the cut after $v_k$. Passing from the cut after
$v_{k-1}$ to the cut after $v_k$ removes the arcs entering $v_k$ and adds
those leaving it; with half-open intervals these are exactly the jobs that
end and start at time $k$. By conservation every such cut weighs $2\theta$,
so together with the unit jobs every slab has active weight $r\theta$, and
no job is active outside $[0,L)$. A peak of at most $\theta$ therefore puts
exactly $\theta$ on every machine on every slab. Color each arc by the
machine of its job. The unit jobs give each machine a constant load, so
its arc load is constant as well: at time $k$ it changes by the outgoing
minus the incoming weight of its color at $v_k$, which must vanish, and on
$[0,1)$ it is the weight of its color leaving $v_0$, which is therefore at
most $\theta$. Conversely, given such a coloring, put the arcs of each
color $j$ on machine $j$. Their load is the constant weight $f_j\le\theta$
that leaves $v_0$ in color $j$, and $\theta-f_j$ unit jobs fill machine $j$
to exactly $\theta$; this uses $(r-2)\theta$ unit jobs, since the $f_j$ sum
to $2\theta$. Integer weights give integer peaks, so without such a
coloring $\OPT\ge\theta+1$.
\end{proof}

\subsection{Weight-two tracks and gadgets}
All networks below have arc weights $1$ and $2$. A \emph{track} is a
directed path of weight-$2$ arcs through the network; its \emph{color} at a
point is the color of its current arc. A \emph{gadget} is a constant number
of consecutive vertices through which some tracks pass; every other track
crosses a gadget on a single arc. In every coloring in which every vertex
conserves the weight of each color, the following hold.
\begin{labelledlist}
\item[(F1)] \emph{Split and merge.} At a vertex with one weight-$2$ arc on
one side and two weight-$1$ arcs on the other, all three arcs have the same
color, since the weight of any color on the weight-$2$ side is $0$ or $2$.
\item[(F2)] \emph{Mixer.} At a vertex with $h\in\{2,3\}$ incoming and $h$
outgoing weight-$2$ arcs, conservation says exactly that every color occurs
equally often on both sides: the outgoing colors are an arbitrary
permutation of the incoming ones.
\item[(F3)] \emph{Equality} $\mathrm{Eq}(A,B)$ for two tracks $A$ and $B$
(Figure~\ref{fig:equality-gadget}). Split $A$ into unit arcs $a_1,a_2$ and
$B$ into $b_1,b_2$; merge $a_2,b_1$ into a weight-$2$ arc $c$; split $c$
into $c_1,c_2$; merge $a_1,c_1$ into the continuation of $A$ and $c_2,b_2$
into the continuation of $B$. By (F1) the six vertices give $a_1$ and
$a_2$ the color of $A$, $b_1$ and $b_2$ the color of $B$, and $a_2,c,b_1$
one common color. So the gadget can be colored if and only if $A$ and $B$
have the same color, and then both continuations keep it.
\item[(F4)] \emph{Matching}
$\mathrm{Match}(A_1,\ldots,A_h;B_1,\ldots,B_h)$ for distinct target tracks
$A_i$ and reference tracks $B_i$ outside them, $h\in\{2,3\}$; references may
repeat. It is a mixer on the targets followed by
$\mathrm{Eq}(A_1,B_1),\ldots,\mathrm{Eq}(A_h,B_h)$. By (F2) and (F3) it can
be colored if and only if the targets carry the same multiset of colors as
the references; afterwards each $A_i$ has the color of $B_i$, and the
references are unchanged.
\end{labelledlist}
\begin{figure}[htb]
\centering
\begin{tikzpicture}[x=1.07cm,y=.66cm,>=Stealth,font=\small,
  bd/.style={circle,draw,fill=white,inner sep=0pt,minimum size=1.6mm,line width=.45pt},
  gv/.style={circle,draw,fill=white,inner sep=0pt,minimum size=2.5mm,line width=.45pt},
  edge/.style={-{Stealth[length=1.25mm,width=.95mm]},line width=.45pt},
  vl/.style={font=\footnotesize,fill=white,inner sep=.6pt}]

\node[bd] (eb0n0) at (0.000,1.000) {};
\node[bd] (eb0n1) at (0.000,-1.000) {};
\node[bd] (eb1n0) at (2.550,1.500) {};
\node[bd] (eb1n1) at (2.550,0.500) {};
\node[bd] (eb1n2) at (2.550,-0.500) {};
\node[bd] (eb1n3) at (2.550,-1.500) {};
\node[bd] (eb2n0) at (5.100,1.500) {};
\node[bd] (eb2n1) at (5.100,0.000) {};
\node[bd] (eb2n2) at (5.100,-1.500) {};
\node[bd] (eb3n0) at (7.650,1.500) {};
\node[bd] (eb3n1) at (7.650,0.500) {};
\node[bd] (eb3n2) at (7.650,-0.500) {};
\node[bd] (eb3n3) at (7.650,-1.500) {};
\node[bd] (eb4n0) at (10.200,1.000) {};
\node[bd] (eb4n1) at (10.200,-1.000) {};
\node[gv,] (eg0n0) at (1.275,1.000) {};
\draw[edge] (eb0n0)--(eg0n0);
\draw[edge] (eg0n0)--(eb1n0);
\draw[edge] (eg0n0)--(eb1n1);
\node[gv,] (eg0n1) at (1.275,-1.000) {};
\draw[edge] (eb0n1)--(eg0n1);
\draw[edge] (eg0n1)--(eb1n2);
\draw[edge] (eg0n1)--(eb1n3);
\node[font=\scriptsize,align=center] at (1.275,-2.150) {split both};
\node[gv,] (eg1n0) at (3.825,1.500) {};
\draw[edge] (eb1n0)--(eg1n0);
\draw[edge] (eg1n0)--(eb2n0);
\node[gv,fill=black!12,] (eg1n1) at (3.825,0.000) {};
\draw[edge] (eb1n1)--(eg1n1);
\draw[edge] (eb1n2)--(eg1n1);
\draw[edge] (eg1n1)--(eb2n1);
\node[gv,] (eg1n2) at (3.825,-1.500) {};
\draw[edge] (eb1n3)--(eg1n2);
\draw[edge] (eg1n2)--(eb2n2);
\node[font=\scriptsize,align=center] at (3.825,-2.150) {merge: $u+v=2w$};
\node[gv,] (eg2n0) at (6.375,1.500) {};
\draw[edge] (eb2n0)--(eg2n0);
\draw[edge] (eg2n0)--(eb3n0);
\node[gv,] (eg2n1) at (6.375,0.000) {};
\draw[edge] (eb2n1)--(eg2n1);
\draw[edge] (eg2n1)--(eb3n1);
\draw[edge] (eg2n1)--(eb3n2);
\node[gv,] (eg2n2) at (6.375,-1.500) {};
\draw[edge] (eb2n2)--(eg2n2);
\draw[edge] (eg2n2)--(eb3n3);
\node[font=\scriptsize,align=center] at (6.375,-2.150) {split $w$};
\node[gv,] (eg3n0) at (8.925,1.000) {};
\draw[edge] (eb3n0)--(eg3n0);
\draw[edge] (eb3n1)--(eg3n0);
\draw[edge] (eg3n0)--(eb4n0);
\node[gv,] (eg3n1) at (8.925,-1.000) {};
\draw[edge] (eb3n2)--(eg3n1);
\draw[edge] (eb3n3)--(eg3n1);
\draw[edge] (eg3n1)--(eb4n1);
\node[font=\scriptsize,align=center] at (8.925,-2.150) {merge pairs};
\node[vl,above=3pt] at (eb0n0) {$2u$};
\node[vl,above=3pt] at (eb0n1) {$2v$};
\node[vl,above=3pt] at (eb1n0) {$u$};
\node[vl,above=3pt] at (eb1n1) {$u$};
\node[vl,above=3pt] at (eb1n2) {$v$};
\node[vl,above=3pt] at (eb1n3) {$v$};
\node[vl,above=3pt] at (eb2n0) {$u$};
\node[vl,above=3pt] at (eb2n1) {$2w$};
\node[vl,above=3pt] at (eb2n2) {$v$};
\node[vl,above=3pt] at (eb3n0) {$u$};
\node[vl,above=3pt] at (eb3n1) {$w$};
\node[vl,above=3pt] at (eb3n2) {$w$};
\node[vl,above=3pt] at (eb3n3) {$v$};
\node[vl,above=3pt] at (eb4n0) {$2u$};
\node[vl,above=3pt] at (eb4n1) {$2v$};
\end{tikzpicture}
\caption{The equality gadget $\mathrm{Eq}(A,B)$. An arc labelled $2u$ has weight~2 and color~$u$. The unit arcs of the shaded merge come from different tracks, so conservation forces the colors $u$ and $v$ of $A$ and $B$ to agree; the last two merges restore both tracks with their colors.}\label{fig:equality-gadget}
\end{figure}
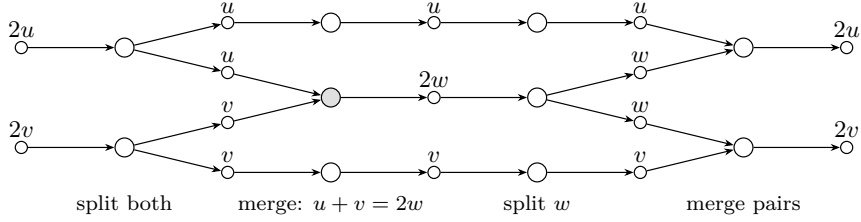

\subsection{A track program for monotone 1-in-3 satisfiability}
A \emph{monotone 1-in-3 formula} $\varphi$ has variables
$1,\ldots,n_\varphi$ and $m_\varphi$ clauses, each a set of three
distinct variables; it is satisfied by a truth assignment in which every
clause contains exactly one true variable.
Its network has one track for each of
\[
 R_{\mathrm T},\ R_{\mathrm F};\qquad
 X_i^0,\ Y_i^0\ (i\in[n_\varphi]);\qquad
 X_i^\gamma,\ Y_i^\gamma\ (\gamma\text{ a clause},\ i\in\gamma),
\]
so the number of tracks is
\begin{equation}\label{eq:track-count}
 \theta=2+2n_\varphi+6m_\varphi .
\end{equation}
All tracks leave $v_0$ on weight-$2$ arcs, so the source emits $2\theta$
and this $\theta$ is the threshold of Lemma~\ref{lem:network}. The
gadgets, which together form the \emph{track program}, follow in three
phases, after which every track enters the sink $v_L$:
\begin{labelledlist}
\item[(i)] $\mathrm{Eq}(X_i^0,X_i^\gamma)$ and $\mathrm{Eq}(Y_i^0,Y_i^\gamma)$
for every clause $\gamma$ and every $i\in\gamma$;
\item[(ii)] $\mathrm{Match}(X_i^0,Y_i^0;R_{\mathrm T},R_{\mathrm F})$ for
every $i\in[n_\varphi]$;
\item[(iii)] $\mathrm{Match}(X_i^\gamma,X_j^\gamma,X_k^\gamma;
R_{\mathrm T},R_{\mathrm F},R_{\mathrm F})$ for every clause
$\gamma=\{i,j,k\}$.
\end{labelledlist}
Every vertex conserves weight, so this is a balanced network. It has
$9m_\varphi+2n_\varphi$ equality gadgets and $n_\varphi+m_\varphi$ mixers,
hence $13n_\varphi+55m_\varphi+2$ vertices and $22n_\varphi+90m_\varphi+2$
arcs, of weights $1$ and $2$.

\begin{lemma}[Track program]\label{lem:track-program}
For every $r\ge2$, the network of a monotone 1-in-3 formula has a
coloring with $r$ colors as in Lemma~\ref{lem:network} if and only if the
formula is satisfiable.
\end{lemma}

\begin{proof}
Suppose such a coloring exists, and compare colors with those in which the
tracks leave $v_0$. Equality gadgets change no color. Phase~(ii) is the only
phase that changes $X_i^0$ or $Y_i^0$, and a clause copy $X_i^\gamma$ is
changed only by the matching of its own clause in phase~(iii). So when
phase~(iii) tests $X_i^\gamma$, it still has its initial color, which by
phase~(i) is the initial color of $X_i^0$. Phase~(ii) makes the initial
colors of $X_i^0,Y_i^0$ the multiset of colors of $R_{\mathrm T},R_{\mathrm F}$.
If $R_{\mathrm T}$ and $R_{\mathrm F}$ had the same color, phases~(i) and~(ii)
would give every track that color, and the arcs of that color leaving
$v_0$ would weigh $2\theta>\theta$. Hence they differ. Call variable $i$
true when $X_i^0$ starts with the color of $R_{\mathrm T}$. Phase~(iii)
then says that every clause contains exactly one true variable.

Conversely, fix a satisfying assignment. Color $R_{\mathrm T}$ red,
$R_{\mathrm F}$ blue, every copy of $X_i$ red exactly when $i$ is true, and
every copy of $Y_i$ with the other of these two colors. One track of every
pair $(X_i^0,Y_i^0)$ or $(X_i^\gamma,Y_i^\gamma)$ is red, so $1+n_\varphi+3m_\varphi=\theta/2$
tracks, of total weight $\theta$, leave $v_0$ red, and the other
$\theta/2$ leave blue. Every equality gadget compares equal colors. The
pairs in phase~(ii), and, since exactly one variable of each clause is
true, the triples in phase~(iii), carry the multisets of colors of their
references. So every mixer can output the colors of its references in
order, and every vertex conserves the weight of each color.
\end{proof}

\subsection{Three-coloring and the lower bounds}
\emph{Three-coloring as a monotone 1-in-3 formula.} For a graph $H$ with
$n_H$ vertices and $m_H$ edges, take variables $X_{u,q}$ for every vertex
$u$ and color $q\in\{0,1,2\}$ and $Z_{uv,q}$ for every edge $uv$ and color
$q$, with the clauses
\[
 \{X_{u,0},X_{u,1},X_{u,2}\}\quad(u\text{ a vertex}),\qquad
 \{X_{u,q},X_{v,q},Z_{uv,q}\}\quad(uv\text{ an edge},\ q\in\{0,1,2\}).
\]
A proper coloring $\chi$ satisfies them with $X_{u,q}$ true exactly when
$\chi(u)=q$, and $Z_{uv,q}$ true exactly when neither $X_{u,q}$ nor
$X_{v,q}$ is. Conversely, the vertex clauses give every vertex one color
and the edge clauses forbid equal colors at the ends of an edge. The
formula $\varphi$ has $n_\varphi=3n_H+3m_H$ variables and
$m_\varphi=n_H+3m_H$ clauses.

\begin{proof}[Proof of Proposition~\ref{prop:lower-bounds}]
\emph{Hardness.} Three-coloring is NP-complete already for graphs of
maximum degree four~\cite{GareyJohnsonStockmeyer1976}. For such a graph
and fixed $r\ge2$, Lemmas~\ref{lem:network} and~\ref{lem:track-program}
applied to its formula give, in polynomial time, an $r$-machine instance
with integer endpoints in $\{0,\ldots,L\}$ and weights in $\{1,2\}$ whose
optimum is $\theta$ if the graph is three-colorable and at least
$\theta+1$ otherwise. An assignment is verified by an endpoint sweep, and
all numbers are polynomially bounded, so the decision problem is strongly
NP-complete~\cite{GareyJohnson1978}.

\emph{No FPTAS.} An FPTAS run with $\varepsilon=1/(2\theta)$ returns, on a
three-colorable input, a peak at most $\theta+1/2$, hence at most $\theta$ by
integrality; on any other input every peak is at least $\theta+1$. Since
$1/\varepsilon=2\theta$ is polynomial in the input, this would decide
three-colorability in polynomial time, so no FPTAS exists unless P$=$NP.

\emph{The ETH bounds.} Maximum degree four gives $m_H\le2n_H$, hence
$n_\varphi\le9n_H$, $m_\varphi\le7n_H$, and $\theta\le60n_H+2$
by~\eqref{eq:track-count}. For fixed $r$ the instance has $O(n_H)$ jobs
with integer endpoints and weights, so $N=O(n_H\log n_H)$. With
$\varepsilon=1/(2\theta)$, a deterministic $(1+\varepsilon)$-approximation
in $2^{o(1/\varepsilon)}\poly(N)$ bit operations would decide
three-colorability of graphs of maximum degree four in
$2^{o(n_H)}\poly(n_H)$ time. Under the Exponential Time
Hypothesis~\cite{ImpagliazzoPaturi2001} no such algorithm
exists~\cite[Lemma~2.2]{Cygan2017}. The randomized Exponential Time
Hypothesis excludes randomized algorithms that decide 3-SAT in time
$2^{o(n)}$ with error probability at most
$1/3$~\cite{DellHusfeldtMarxTaslamanWahlen2014}. The reductions behind
\cite[Lemma~2.2]{Cygan2017} are deterministic, so it excludes such
algorithms for three-coloring graphs of maximum degree four as well. A
randomized $(1+\varepsilon)$-approximation that succeeds with probability
at least $2/3$ would give one: on a three-colorable input it returns, with
probability at least $2/3$, an assignment of peak at most $\theta+1/2$ and
hence $\theta$, no other input has such an assignment, and a returned
assignment is checked by an endpoint sweep. All of this holds for every
fixed $r\ge2$.
\end{proof}

\end{document}